\documentclass{llncs}
\usepackage{rotating}
\usepackage{graphicx}
\usepackage{hyperref}
\usepackage{amsmath}
\usepackage{algorithmicx}
\usepackage{algorithm}
\usepackage{amsfonts}
\usepackage{algpseudocode}
\usepackage{multirow}
\usepackage{array}
\usepackage[normalem]{ulem}
\usepackage{thmtools}
\usepackage{caption}
\usepackage{subcaption}
\newcount\Comments  
\usepackage{color}
\newcommand{\kibitz}[2]{\ifnum\Comments=1\textcolor{#1}{#2}\newline\fi}

\algtext*{Until}

\usepackage[left=3cm,right=3cm,top=3cm,bottom=3cm]{geometry}
\mainmatter
\title{How Best to Handle a Dicey Situation}
\author{Rishab Nithyanand \and Jonathan Toohill \and Rob Johnson}
\institute{Department of Computer Science\\Stony Brook University\\
\email{\{rnithyanand, rob\}@cs.stonybrook.edu, jonathan.toohill@stonybrook.edu}}
\begin{document}
\maketitle              

\begin{abstract}
We introduce the {Destructive Object Handling} (DOH) problem, which models aspects of many real-world allocation problems, such as shipping explosive munitions, scheduling processes
in a cluster with fragile nodes, re-using passwords across multiple
websites, and quarantining patients during a disease outbreak.  In
these problems, objects must be assigned to handlers, but each object
has a probability of destroying itself and all the other objects
allocated to the same handler.  The goal is to maximize the expected
value of the objects handled successfully.


We show that finding the optimal allocation is
$\mathsf{NP}$-$\mathsf{complete}$, even if all the handlers are
identical. We present an FPTAS when the number of handlers is
constant.  We note in passing that the same technique also yields a first
FPTAS for the weapons-target allocation problem~\cite{manne_wta} with
a constant number of targets.  We study the structure of DOH problems
and find that they have a sort of phase transition -- in some instances
it is better to spread risk evenly among the handlers, in others, one
handler should be used as a ``sacrificial lamb''.  We show that the
problem is solvable in polynomial time if the destruction
probabilities depend only on the handler to which an object is
assigned; if all the handlers are identical and the objects all have
the same value; or if each handler can be assigned at most one object.
Finally, we empirically evaluate several heuristics based on a combination of greedy and genetic algorithms. The proposed heuristics return fairly high quality solutions to very large problem instances (upto 250 objects and 100 handlers) in tens of seconds.
\end{abstract}

\section{Introduction}\label{sec:introduction}
Consider the problem of transporting multiple types of explosives
across a war-zone when multiple options of transportation are
available (eg., armoured trucks, naval carriers, submarines, transport
aircrafts, etc). Each transportation option has a non-negligible
probability of being destroyed by enemy combatants or by the
mishandling of (possibly faulty) explosives. The goal is to find an
allocation of explosives to transportation units such that expected
value of delivered explosives is maximized. We classify such
allocation problems as {Destructive Object Handling} ({DOH})
problems.
 
Formally, a {DOH} problem consists of a set of $n$ objects that
must be allocated to $k$ handlers.  For $j=1,\dots, n$, the $j^{th}$
object has value $v_j$ and a vector of non-self-destruction
probabilities $\langle p_{1j}, \dots, p_{kj} \rangle$. Here, $p_{ij}$
denotes the probability that the $j^{th}$ object will \emph{not}
self-destruct if allocated to the $i^{th}$ handler.  All the objects
allocated to a handler survive if none of them self-destruct,
otherwise they are all lost.  Self-destruction probabilities are
assumed to be independent.  The goal is to find an allocation that
maximizes the expected value of surviving objects.

\begin{problem}\textit{The Generic Destructive Object Handling ({DOH}) problem.}
  Given $v_1,\dots,v_n$ and $p_{ij}$ for $i=1,\dots k$ and
  $j=1\dots,n$, find partition $S_1\cup\dots\cup S_k$ of
  $\{1,\dots,n\}$ that maximizes $ES=\sum_{i=1}^kES_i$,
  where $ES_i=\sum_{j\in S_i}v_j\prod_{j\in S_i} p_{ij}$.
\label{problem:DOH}
\end{problem}
 We will use the notation
$V_i=\sum_{j\in S_i}v_j$ and $P_i= \prod_{j\in S_i} p_{ij}$, so
$ES=\sum_{i=1}^kES_i=\sum_{i=1}^kP_iV_i$.
Note that the DOH problem can also model values for handlers and
allocation-independent destruction probabilities by creating proxy
objects corresponding to each handler. In addition to hazardous object shipping, {DOH} problems arise in many other real world contexts.  

\emph{Emergency quarantine.}  During an outbreak of a disease 
with a short incubation period, medical staff may need to quarantine
patients collectively, i.e. patients may be sealed in
multiple rooms for the incubation duration.  If no one in a room is
infected, then all the patients in that room survive and are released.
If any one of the patients in a room is sick, then all the
patients in that room are, sadly, lost. Given an initial
estimate of the probability that a patient is sick based on physical
observations, such as fever, cuts, bite marks, etc., how should
staff assign patients to rooms to maximize the expected number
survival rate?  

\emph{Process allocation with marginal nodes.} Large computing 
clusters may have nodes that are marginal, i.e.  they sometimes crash
due to faulty RAM, overheating, etc., and the probability that
a node crashes may depend on the number and kind of processes
allocated it.  Assuming no checkpointing, if a node
crashes, then all the work it has performed is lost. Given a value
for each task, and probabilities $q_{ij}$ that task $j$ will crash
machine $i$, how should processes be allocated to machines so as to
maximize the expected value of completed tasks?


\emph{Re-using passwords.}  Recent studies show that
humans, on average, distribute six passwords amongst 25 web service
accounts~\cite{web_passwords_1}.  If an attacker manages to steal the
password database for one of these services~\cite{real-breakins} and
recovers the user's password for that service, then, due to the strong
linkability \cite{web_passwords_2} of web based accounts, he can
easily break into all the other accounts that share that password.  If
the user's $i^{th}$ password is assigned to her $j^{th}$ account, then the
probability, $q_{ij}$, that an attacker obtains the $i^{th}$ password by
breaking into the $j^{th}$ web service is a function of (1) the security
of the $j^{th}$ web service, and (2) the strength of the $i^{th}$ the
password (but only if the $j^{th}$ service stores password hashes instead
of passwords~\cite{no-hash}).  Thus, given estimates of the value of
each account, the security of each service provider, the strength of
each password, and whether each service provider stores passwords or
hashes, we can ask: how should a user allocate accounts to passwords
so as to minimize her expected loss from server breakins?

\begin{table}[t]
\centering
\begin{tabular}{ |m{4.4cm}|>{\centering\arraybackslash}m{1.9cm}|>{\centering\arraybackslash}m{2.1cm}|>{\centering\arraybackslash}m{3.2cm}|}
  \hline
                                                                                                                                  &                                              & \multicolumn{2}{c|}{\textbf{Approximation Results}}                                   \\
  \cline{3-4}
  \centering{\textbf{Problems}}                                                                                                   & \textbf{Class}                               & \emph{Factor}               & \emph{Time Complexity}                                  \\
  \cline{1-4}
  General DOH problem (Sections~\ref{sec:complexity}, \ref{subsec:approx:general})                                                 & at least $weakly$ $\mathsf{NP}$-$\mathsf{C}$ & $\frac{(1-\epsilon)(c-1)}{k}$       & $O(\frac{k^cn^{2c+1}\ln^c P \ln^c V}{\epsilon^{2c}})$ \\
  \cline{1-4}
  DOH with $k=O(1)$ (Sections~\ref{sec:complexity}, \ref{sec:approximations})                                                           & $weakly$ $\mathsf{NP}$-$\mathsf{C}$          & $1-\epsilon$                & $O(\frac{n^{2k+1} \ln ^{k}V\ln ^{k}P}{\epsilon^{2k}})$  \\
  \cline{1-4}
  Identical Handlers  (Section~\ref{sec:complexity})                                                   & at least $weakly$ $\mathsf{NP}$-$\mathsf{C}$ & Same as General DOH & Same as General DOH \\
  \cline{1-4}
  Identical Values                                                                                  & \emph{Open}                                  & Same as General DOH & Same as General DOH \\
  \cline{1-4}
  Identical Objects (Section \ref{subsec:special_cases:identical_objects})                     & $\mathsf{P}$                                 & 1                           & $O(nk)$                                                 \\
  \cline{1-4}
  Identical Handlers and Values (Section~\ref{subsec:special_cases:identical_handlers_values}) & $\mathsf{P}$                                 & 1                           & $O(n^2k)$                                               \\
  \cline{1-4}
  Identical Risks (Section~\ref{subsec:special_cases:identical_risks})                                & $\mathsf{P}$                                 & 1                           & $O(n^2k)$                                               \\
  \cline{1-4}
  At most one object per handler (Section \ref{subsec:special_cases:EOVA})                                                        & $\mathsf{P}$                                 & 1                           & $O\left(n^3k \log n\right)$                             \\
  \hline
  \hline
  \cline{1-4}
\end{tabular}
\caption{Summary of Complexity Results.
  Here, $k$ is the number of handlers, $n$ is the number of objects, $P=\min(\prod_{j = 1}^n p_{1j},
  \dots, \prod_{j = 1}^n p_{kj})$, and $V=\min_j v_j/\sum_{j=1}^nv_j$.
  \vspace{-.3in}}
\label{tab:summary}
\end{table}

\subsection{Contributions and Organization}
A summary of theoretical results presented in this paper is shown in
Table~\ref{tab:summary}.  In Section~\ref{sec:complexity}, we show
that the DOH problem is NP-complete, even if all the handlers are
identical.  We provide an FPTAS for the DOH problem
when the number of handlers is a constant in Section \ref{sec:approximations}. Coincidentally, the same technique also yields the first FPTAS for the Weapons-Target Allocation problem with a constant
number of targets.  In
Section~\ref{sec:special_cases}, we provide polynomial time algorithms
for the following special cases: when the survival probabilities
depend only on the handler, when all objects have the same values and
all handlers are identical, when all the objects are identical, and
when at most one object can be allocated to each handler.  Heuristic
approaches for finding allocations are proposed and compared in
section \ref{sec:heuristics}. Finally, we summarize our conclusions.

\subsection{Related Work}\label{subsec:introduction:related}
In terms of similarity of objective functions, the {DOH} problem
is most closely related to the {static Weapon-Target Allocation
  problem (WTA)}. The {WTA} problem \cite{manne_wta} is a well studied non-linear allocation problem in
the field of command-and-control theory in which the objective is to
allocate missiles to enemy locations so as to inflict maximum damage
(given that each missile $j$ destroys a target $i$ with probability
$p_{ij}$). The problem was shown to be at least
{weakly} $\mathsf{NP}$-$\mathsf{Complete}$ by Lloyd and
Witsenhausen \cite{Witsenhausen1986WTA} \footnote{To date it remains unknown if
the problem is {strongly} $\mathsf{NP}$-$\mathsf{Complete}$ as
there exist no known strong reductions, pseudo-polynomial time algorithms, or
{FPTAS}.}. We show in this paper that, when the number of targets
is constant, the problem admits an FPTAS. 
While certain similarities exist between the {WTA} objective
function ($F_{wta} = \sum_{i = 1}^k \prod_{j \in S_i}
p_{ij}$)\footnote{$S_i$ denotes the set of missiles allocated to the $i^{th}$ target.} and the {DOH} objective function
(Problem \ref{problem:DOH}), a major complicating difference is the absence
of an additive sub-component to $F_{wta}$. As a result, unlike the
{WTA} problem, the objective function of the {DOH} problem
does not yield a convex function even after relaxation of the integer
requirement.

The {DOH} problem has many applications in the domain of
{hazardous material processing and routing}~\cite{hazmat}. These
problems generally deal with the selection of {minimum risk}
processing locations \cite{hazmat_fl} and transportation routes in
networks \cite{hazmat_1}, \cite{hazmat_2} -- i.e., finding minimum
cost facility locations and routes for which human and external
material loss in the event of malfunction incidents is minimized.

\section{Complexity of {DOH} Problems}\label{sec:complexity}
In this section, we show that the DOH problem is at least weakly
$\mathsf{NP}$-$\mathsf{Complete}$.

\begin{definition}\label{definition:DVA}
Decisional {DOH} Problem (D-DOH): Given a DOH problem and a threshold
$r$, does there exist an allocation such that $ES\geq r$?
\end{definition}

\begin{theorem}\label{th:np_completeness}
$D$-{DOH} problems are at-least weakly $\mathsf{NP}$-$\mathsf{Complete}$.
\end{theorem}
\begin{proof}
Membership in $\mathsf{NP}$ is easily established. Given a guess for
$S_1, \dots, S_k$, we just compute $ES$ and verify that it is at least
$r$. This can be done in time proportional to the length of the words
representing the values and probabilities.

We will reduce the well studied $3$-$Partition$ problem ($3P$) as defined in
\cite{3_Part}, to the $D$-{DOH} problem. Given an instance
$I_{3P}: X = \{x_1, \dots, x_{3n}\}$ (where $\sum_{j=1}^{3n} x_j =
nB$) of the $3P$ problem, we create an instance $I_{DOH}$ of the
$D$-{DOH} problem with $3n$ objects and $n$ handlers as follows:
\begin{itemize}
\item Pick a rational number $b\in (1, e^{4/nB})$.
\item Set $v_j = x_j$, $\forall j \in \{1, \dots, 3n\}$.
\item Set $p_{ij} = b^{-x_j}$, $\forall j \in \{1, \dots, 3n\}$ and $\forall i \in \{1, \dots, n\}$.
\item Set $r = n\times B \times b^{-B}$.
\end{itemize}
By construction, $\sum_{i=1}^nV_i=nB$ and, for any allocation
$S_1,\ldots, S_n$, we have $ES=\sum_{i=1}^nV_ib^{-V_i}$.
We now argue that, because of the choice of the base, $b$, $ES\leq
\sum_{i=1}^nBb^{-B}=r$.  Consider any unequal allocation, i.e. in
which there exist $i$ and $i'$ such that $V_i\not=V_{i'}$.  We will
show that the allocation would be improved by redistributing the value
equally between handlers $i$ and $i'$.  Consider the function
$f(x)=xb^{-x}+(c-x)b^{-(c-x)}$, where $c=V_i+V_{i'}\leq nB$.  First,
observe that $f'(c/2)=0$.  Second, since $b<e^{4/c}$, $f'(x) > 0$ for
all $x\in[0,c/2)$ and $f'(x)<0$ for $x\in(c/2,c]$.  Hence $x=c/2$ is a
global maximum.  Thus, if an allocation has any $i,i'$ such that
$V_i\not=V_{i'}$, we could improve it by replacing $V_i$ and $V_{i'}$
by $(V_i+V_{i'})/2$.  Hence the optimal allocation must have $ES\leq
\sum_{i=1}^nBb^{-B}=r$, and this can only be obtained if
$V_1=\dots=V_n=B$.

Therefore, if there is a solution to the 3P problem, then there exists an
allocation for the D-DOH problem in which $V_1=\dots=V_n=B$, so that
$ES=r$.  If, on the other hand, there is an allocation of the D-DOH
problem such that $ES\geq r$, then we must have $ES=r$ and
$V_1=\dots=V_n=B$, so there exists a solution to the original 3P
problem.  
\end{proof}
Theorem \ref{th:np_completeness} shows that even the
restricted {D-DOH} problem -- where probabilities
are object dependent but handler independent -- is weakly
$\mathsf{NP}$-$\mathsf{Complete}$. It remains an open question
whether the {D-DOH} problem is strongly $\mathsf{NP}$-$\mathsf{Complete}$.

\section{Boundable Approximations}\label{sec:approximations}

We now derive an approximation algorithm for the DOH problem.  The
algorithm, which is based on dynamic programming with state-space
trimming~\cite{fptas_woeginger}, is an FPTAS when $k=O(1)$. We then build on this result and present an approximation for arbitrary $k$ in Section \ref{subsec:approx:general}.

Consider allocating the objects to handlers one at a time, i.e. we
allocate the first object, then the second, etc.  Let $S_{ti}$ be the
set of objects allocated to the $i^{th}$ handler at time step $t$,
$V_{ti}=\sum_{j\in S_{ti}}v_j$, and $P_{ti}=\prod_{j\in S_{ti}}p_{ij}$,
$ES_{ti}=P_{ti}V_{ti}$, and
$ES_t=\sum_{i=1}^kES_{ti}=\sum_{i=1}^kP_{ti}V_{ti}$.  If we allocate
the $t+1$st object to the $\ell$th handler, then we will have
\[
\begin{array}{lcr}
  P_{t+1,i}=\left\{
  \begin{array}{ll}
    P_{t,i}p_{\ell,t+1} & \textrm{ if }i=\ell \\
    P_{t,i} & \textrm{ if }i\not=\ell
  \end{array}
  \right.
  &\;\;\;\;\;\;\;\;\; &
  V_{t+1,i}=\left\{
  \begin{array}{ll}
    V_{t,i}+v_{\ell,t+1} & \textrm{ if }i=\ell \\
    V_{t,i} & \textrm{ if }i\not=\ell
  \end{array}
  \right.
\end{array}
\]
Thus $(V_{t1}, \dots, V_{tk}, P_{t1},\dots,P_{tk})$ is the only state
information we need to compute the state after allocating the $t$'th
object. This also gives a dominance relation among allocations: if $V_{ti}\leq
V'_{ti}$ and $P_{ti}\leq P'_{ti}$ for all $i$, then every extension of
allocation $(V_{t1},\dots,V_{tk},P_{t1},\dots,P_{tk})$ will have lower
ES than the corresponding extension of
$(V'_{t1},\dots,V'_{tk},P'_{t1},\dots,P'_{tk})$.  Thus we only need to
consider $(V'_{t1},\dots,V'_{tk},P'_{t1},\dots,P'_{tk})$ in our search
for the optimal allocation.

\begin{algorithm}[h]
  \caption{Dynamic Program for Exact Solutions}
  \label{alg:dp_exact}
  \begin{algorithmic}
    \Function{DOH-Solve}{$v_1, \dots, v_n, p_{11}, \dots, p_{kn}$}
      \State $\Phi_0 \gets \{(0, \dots, 0, 1, \dots, 1)\}$
      \For{$j =1 \to n$}
        \State $\Phi_j \gets \emptyset$
	      \ForAll{$(V_1,\dots,V_k,P_1,\dots,P_k) \in \Phi_{j-1}$}
          \ForAll{$i=1\to k$}
            \State  $\begin{array}{lcr}
              V'_\ell=\left\{
              \begin{array}{ll}
                V_i+v_{ij} & \textrm{ if }i=\ell \\
                V_\ell & \textrm{ if }i\not=\ell
              \end{array}
              \right.
              & \;\;\;\;\;\;\;\;\;\;\; &
              P'_\ell=\left\{
              \begin{array}{ll}
                P_ip_{ij} & \textrm{ if }i=\ell \\
                P_\ell & \textrm{ if }i\not=\ell
              \end{array}
              \right.
            \end{array}$
 	          \State $\Phi_j \gets \Phi_j \cup \{(V'_1, \dots, V'_k, P'_1, \dots, P'_k)\}$
          \EndFor
        \EndFor
        \State $\Phi_j \gets$ \Call{Trim}{$\Phi_j$}
      \EndFor
      \State \Return $\max_{D\in\Phi_n}ES(D)$
    \EndFunction
    \Function{Trim}{$\Phi$} \Comment{\textsc{Trim} function for computing exact solutions}
      \State \Return $\Phi$
    \EndFunction
  \end{algorithmic}
\end{algorithm}



\subsection{An FPTAS for constant $k$}\label{subsec:approximations:fptas}
We derive an FPTAS for the {DOH} problem with $k = O(1)$ using an
approach based on {trimming the state space}
\cite{fptas_woeginger} and the dynamic program shown in algorithm
\ref{alg:dp_exact}. The basic idea is to reduce the size of the state
space from exponentially to polynomially large by collapsing
{similar} states into a single state. We introduce a
{trimming parameter} $\delta = 1 + \frac{\epsilon}{4kn}$. Here,
$1 - \epsilon$ is the desired approximation, with $0 < \epsilon < 1$.

During the execution of the algorithm, we divide the state space into
$2k$-orthotopes whose boundaries are of the form
$[\delta^{r+1},\delta^r)$, and trim our states to keep track of at
  most one allocation that falls in each orthotope, as shown in
  Algorithm~\ref{alg:fptas}.

\begin{algorithm}
  \caption{FPTAS trimming procedure}
  \label{alg:fptas}
  \begin{algorithmic}
    \Function{Trim}{$\Phi$} \Comment{\textsc{Trim} function for computing approximate solutions}
	    \ForAll{$(V_1,\dots,V_k,P_1,\dots,P_k) \in \Phi$}
	      \State Find $\langle {r_1}, \dots, {r_{k}} \rangle$ such that $\delta^{r_i+1}\leq V_i \leq \delta^{r_i}$, $\forall i \in \{1, \dots, k\}$
	      \State Find $\langle s_{1}, \dots, {s_{k}} \rangle$ such that $\delta^{s_{i}+1}\leq P_i \leq \delta^{s_{i}}$, $\forall i \in \{1, \dots, k\}$
	      \If{$Orthotope[\langle r_1, \dots, r_{k}, s_1,\dots,s_k \rangle] = \perp$}
		      \State $Orthotope[\langle r_1, \dots, r_{k}, s_1,\dots,s_k \rangle] \gets (V_1,\dots,V_k,P_1,\dots,P_k)$
	      \EndIf
      \EndFor
      \State \Return $Orthotope$
    \EndFunction
  \end{algorithmic}
\end{algorithm}

\begin{theorem}\label{th:fptas}
If $k = O(1)$, then Algorithm~\ref{alg:dp_exact} with the trimming
function in Algorithm~\ref{alg:fptas} is a fully polynomial time
approximation scheme (FPTAS).
\end{theorem}
\begin{proof}
We must bound the number of states examined by the algorithm.  Since,
during each step of the allocation, it maintains at most one state per
orthotope, we can bound the number of states by bounding the number of
orthotopes.  For clarity, we assume in this proof that $\sum_j v_j=1$.

The smallest $V_i$ value that may occur during the execution of
Algorithm~\ref{alg:dp_exact} is $V_{min} = min(v_1, \dots, v_n)$.
Thus every $V_i$ value will always fall into one of the $\mid L_v
\mid$ intervals ($\left[\delta^{- L_v}, \delta^{- L_v + 1}\right),
  \dots, \left[\delta^{-2}, \delta^{-1}\right), \left[\delta^{-1},
      1\right]$), where $\delta^{-L_v} \leq V_{min} \leq \delta^{-L_v
      + 1}$ -- i.e., $L_v = \lfloor \frac{\ln V_{min}}{\ln \delta}
    \rfloor$.
Similarly, the smallest $P_i$ value that may occur during the
execution of Algorithm~\ref{alg:dp_exact} is $P_{min} = min(\prod_{j =
  1}^n p_{1j}, \dots, \prod_{j = 1}^n p_{kj})$.  Thus every $P_i$ value
will fall into one of $\mid L_p \mid$ intervals ($\left[\delta^{-L_p},
  \delta^{-L_p + 1}\right), \dots, \left[\delta^{-1}, 1\right]$), where
    $\delta^{-L_p} \leq P_{min} \leq \delta^{-L_p + 1}$ -- i.e., $L_p
    = \lfloor \frac{\ln P_{min}}{\ln \delta} \rfloor$.

Therefore, all states will fall within one of $L_v^kL_p^k$ orthotopes.  In
each iteration, the algorithm may generate $k$ new states for each
state in $\Phi_{j-1}$, so the total number of states generated during
all iterations is bounded by $nkL_v^kL_p^k$.  We know that $L_p^k =
\lfloor \frac{\ln P_{min}}{\ln \delta} \rfloor ^k \leq \left(
\frac{1+\frac{\epsilon}{4kn} \ln P_{min}}{\frac{\epsilon}{4kn}}
\right)^k \leq \left( \frac{8kn \ln P_{min}}{\epsilon}
\right)^k$. Similarly, we have $L_v^k \leq \left( \frac{8kn \ln
  V_{min}}{\epsilon} \right)^k$. Thus the running time of the
algorithm is $O(\frac{(nk)^{2k+1} \ln ^{k}V_{min}\ln
  ^{k}P_{min}}{\epsilon^{2k}})$.  If $k=O(1)$, this simplifies to
$O(\frac{n^{2k+1} \ln ^{k}V_{min}\ln ^{k}P_{min}}{\epsilon^{2k}})$.
Note that, if $N$ is the number of bits used to represent the problem
in base 2, then $\ln V=O(N)$, $\ln P=O(N)$, and $n=O(N)$, so the
running time is polynomial in $N$ and $1/\epsilon$.

Let $\Phi_0,\dots,\Phi_n$ be the sequence of state sets computed by
Algorithm~\ref{alg:dp_exact} without state-space trimming, and
$\Phi'_0,\dots,\Phi'_n$ the sequence of state sets computed by
Algorithm~\ref{alg:dp_exact} with state-space trimming.  By induction
on $i$, it is easy to see that for any state $D \in \Phi_i$, there is
a corresponding state $D' \in \Phi'_i$ such that
$\frac{ES[D]}{(1+\frac{\epsilon}{4nk})^{2ki}} \leq ES[D'] \leq ES[D]$,
which implies that $\frac{ES[D']}{ES[D]}\geq
\frac{1}{(1+\frac{\epsilon}{4nk})^{2ki}} \geq (1 - \epsilon)$.

Since we have shown that, if $k=O(1)$, our algorithm
executes in time bounded by a polynomial in $n$ and
$\frac{1}{\epsilon}$ and finds a solution within a factor $(1 -
\epsilon)$ of the optimal, it follows that it
is an {FPTAS} for the DOH problem when $k=O(1)$.
\end{proof}

\subsection{An Approximation for Arbitrary $k$}\label{subsec:approx:general}

We now derive an approximation algorithm for general $k$ by analyzing
the ratio between the optimal solution using all $k$ handlers and
optimal solutions that use only $c$ handlers.  Thus, by solving the
problem for $c$ handlers, which can be done in polynomial time for
constant $c$, we can obtain a $(1-\epsilon)(c-1)/k$ approximation for
the general problem.

\begin{theorem}
  For a DOH problem with $k$ handlers, let $ES$ be the expected
  surviving value from the optimal allocation.  For any subset $H$ of
  handlers, let $ES_H$ be the expected surviving value from the
  optimal allocation using only those handlers.  For any constant
  $c\geq 2$, there exists an $H$ with $|H|=c$ such that $ES/ES_H\leq
  k/(c-1)$.
\end{theorem}
\begin{proof}
  Let $ES_i$ be the expected surviving value of the $i^{th}$ handler in
  the optimal solution using $k$ handlers.  Assume \emph{w.l.o.g.} that $ES_1
  \geq \dots \geq ES_k$.  If we reallocate all the objects assigned to
  handlers $c+1, \dots,k$ to handler $c$, then the resulting
  allocation will use exactly $c$ handlers (i.e. $H=\{1,\dots,c\}$,
  and will have 
  $ES_H\geq ES_1+\dots+ES_{c-1}\geq \frac{c-1}{k}ES$.
\end{proof}

To obtain a $(1-\epsilon)(c-1)/k$ approximation, we run the FPTAS from
Section~\ref{subsec:approximations:fptas} for every possible subset, $H$, of
$c$ handlers, and take the maximum $ES_H$ value.  Since there are
$\binom{k}{c}=O(k^c)$ such subsets, the running time will be
$O(k^cn^{2c+1}\ln^cV\ln^cP\epsilon^{-2c})$ and the
approximation factor will be $(1-\epsilon)(c-1)/k$.

\section{Special Cases}\label{sec:special_cases}
\subsection{The Case of Identical Objects}\label{subsec:special_cases:identical_objects}
In this section, we study the case where each object has the same value and
possesses failure probabilities that are handler dependent (but,
object independent). 

Since all objects are identical, \emph{w.l.o.g.,} we have $v_j = 1$ and
$p_{ij} = p_i$ ($\forall i \in \{1, \dots, k\}$ and $\forall j \in
\{1, \dots, n\}$). Now the problem can be stated as:
\begin{problem}\label{problem:hdp_DOH}
The Identical Objects DOH (IO-DOH) problem: 
Given $p_{ij} = p_i$ and $v_j = 1$ for $i=1,\dots k$ and
  $j=1\dots,n$, find partition $S_1\cup\dots\cup S_k$ of
  $\{1,\dots,n\}$ that maximizes $ES=\sum_{i=1}^kES_i$,
  where $ES_i=|S_i|p_i^{|S_i|}$.
\end{problem}

\begin{figure}
        \centering
        \begin{subfigure}[b]{0.495\textwidth}
                \centering
                \includegraphics[width=\textwidth]{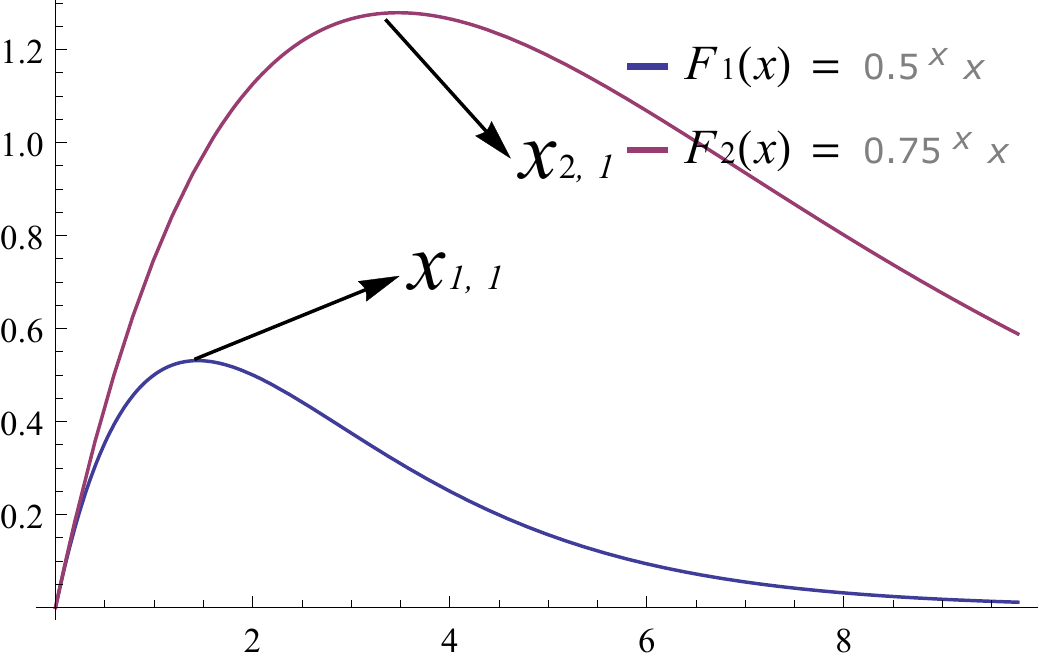}
                \caption{Plot of $xp_i^x$ for $p_1 = .5$, $p_2 = .75$}
                \label{fig:Fx}
        \end{subfigure}%
        ~ 
        \begin{subfigure}[b]{0.495\textwidth}
                \centering
                \includegraphics[width=\textwidth]{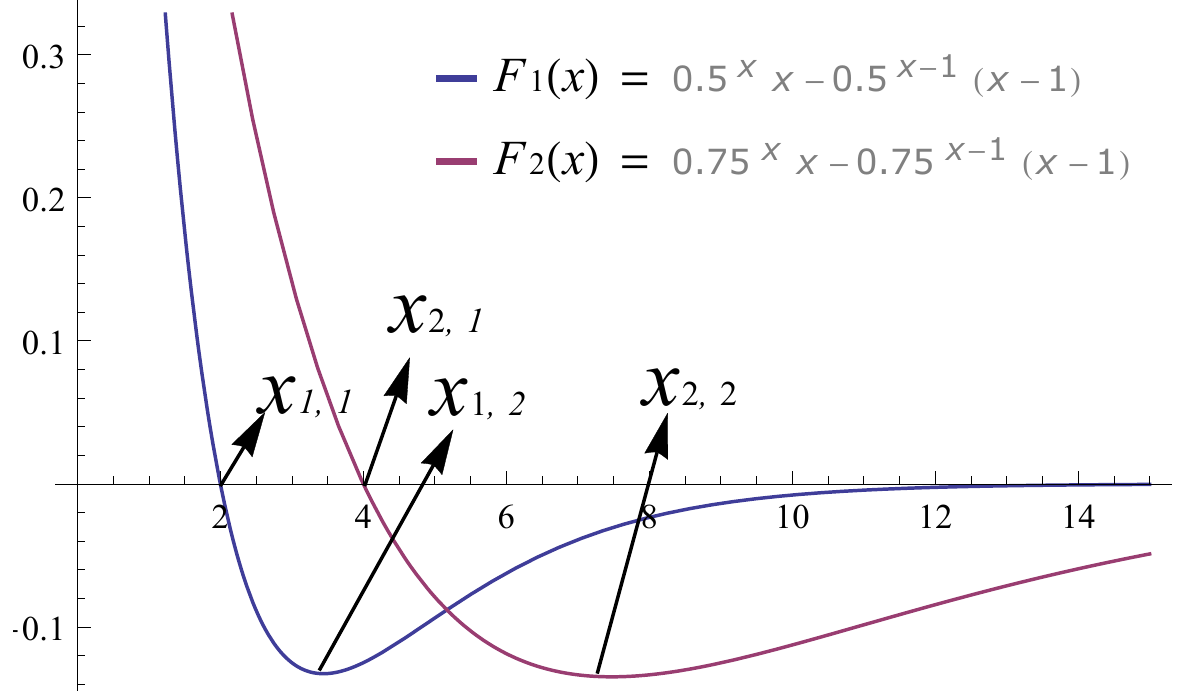}
                \caption{Gradient of $xp_i^x$ for $p_1 = .5$, $p_2 = .75$}
                \label{fig:F1x}
        \end{subfigure}
        \caption{Characteristics of the IO-DOH Problem: For any $0< p_i < 1$, there exist inflection points $x_{i,1}$ and $x_{i,2}$ for the function $F_i(x) = xp_i^x$ and its gradient, respectively. For the function $F_i(x) = xp_i^x$, the region from $x = 0,\dots, x_{i,1}$ is convex, while the region beyond $x = x_{i,1}$ is not. Further, the inflection points ($x_{i,1}$ and $x_{i,2}$) and their corresponding values increase with an increase in $p_i$.}
        \label{fig:identical_objects}
\end{figure}


While a simple application of Jensens inequality \cite{inequalities} is sufficient to prove that the optimal allocation cannot be one where objects are split evenly, we are able to understand substantially more about the problem structure through the exchange arguments made to prove Lemmas \ref{lemma:MMR} and \ref{lemma:runaway} (resulting in Algorithms \ref{alg:mmr} and \ref{alg:hdp_DOH}, respectively).

The Maximum Marginal Return algorithm \cite{MMR_1} is a greedy
approach which makes allocations of objects to the handlers that cause
the greatest increase in the objective function. Lemma \ref{lemma:MMR}
shows that this technique always finds an optimal solution when $n
\leq \sum_{i=1}^k x_{i,1}$. However, when $n > \sum_{i=1}^k x_{i,1}$, MMR is not optimal due to the absence of a convex structure as shown in Figure \ref{fig:identical_objects}. In this case, Lemma \ref{lemma:runaway} shows that an optimal solution has one {sacrificial} handler -- i.e., it is beneficial to allow the highly probable destruction of one handler.
%

\begin{lemma}\label{lemma:MMR}
The Maximum Marginal Return (MMR) algorithm returns an optimal
solution when $n\leq \sum_{i=1}^kx_{i,1}$.
\end{lemma}
\begin{proof}

The following proof is via induction and an exchange argument that exploits the presence of the convex region of $F(x)$, as illustrated in Figure \ref{fig:identical_objects}. 

Clearly, MMR finds the optimal
allocation when only one object is to be allocated. Since all objects are identical, we represent the allocation $S = \langle S_1, \dots, S_k \rangle$ of $n$ objects to $k$ handlers by $O_{n,k} = \langle o_1, \dots, o_k \rangle$, where $o_i = |S_i|$.

We assume that the
solution returned by MMR for allocating $n$ objects to $k$ handlers --
$O_{n,k} = \langle o_1, \dots, o_k \rangle$ -- is optimal. Let handler
$i$ be the handler for which $F_i(o_i+1) - F_i(o_i)$ is maximum --
i.e., allocating a object to this handler causes the largest increase
(or, smallest decrease) in the expected survival value. According to
the MMR algorithm, the $n+1 ^{th}$ object is allocated to this handler
-- i.e., we have $OPT^*_{n+1,k} = \langle o^*_1, \dots, o^*_i, \dots,
o^*_k \rangle = \langle o_1, \dots, o_i+1, \dots, o_k \rangle$. The
change in expected survival value for going from $OPT_{n,k}$ to
$OPT^*_{n+1,k}$ is:
\begin{equation}\label{eq:MMR1}
F'_i(o_i) = F_i(o_i+1) - F_i(o_i)
\end{equation}

Let $Z_{n+1,k} = \langle z_1, \dots z_k \rangle$ be any {other}
solution for allocating $n+1$ objects to $k$ handlers. Clearly, no
optimal solution will have a handler $i$ with more than $x_{i,1}$
objects allocated to them (due to the fact that $x = x_{i,1}$ is a global maximum and strictly decreasing for all greater values of $x$) -- therefore we will
assume $z_i \leq x_{i,1}$ ($\forall i \in \{1,...,k\}$). Note that
since $Z_{n+1,k} \neq OPT^*_{n+1,k}$, there must be some handler $j$
for which:
\begin{equation}\label{eq:MMR2}
z_j > o^*_j \geq o_j
\end{equation}
From this handler we remove a single object (in $Z_{n+1,k}$), to get
an allocation given by $Z^*_{n,k} = \langle z^*_1, \dots, z^*_j,
\dots, z^*_k\rangle = \langle z_1, \dots, z_j - 1, \dots,
z_k\rangle$. The change in expected survival value for going from
$Z^*_{n,k}$ to $Z_{n+1,k}$ is:
\begin{equation}\label{eq:MMR3}
F'_j(z^*_j) = F_j(z_j) - F_j(z_j-1)
\end{equation}

Now, in $OPT_{n,k}$, we add a object to the same handler $j$ to get
$\widehat{OPT}_{n+1,k} = \langle \widehat{o_1}, \dots, \widehat{o_j},
\dots, \widehat{o_k}\rangle = \langle o_1, \dots, o_j + 1, \dots, o_k
\rangle$. The change in expected survival value for going from
$OPT_{n,k}$ to $\widehat{OPT}_{n+1,k}$ is:
\begin{equation}\label{eq:MMR4}
F'_j(o_j) = F_j(o_j+1) - F_j(o_j)
\end{equation}

Due to the assumption that $OPT_{n,k}$ is optimal, (\ref{eq:MMR2}),
and the fact that $x = x_{i,1}$ is an inflection point with strictly decreasing gradient until $x = x_{i,2}$ (see Figure \ref{fig:identical_objects}), we have that (\ref{eq:MMR1}) $\geq$ (\ref{eq:MMR4}) $\geq$
(\ref{eq:MMR3}). Therefore, the solution returned by MMR is at least
as good as any other solution -- i.e., $OPT^*_{n+1,k} \geq Z_{n+1,k}$.
\end{proof}

\begin{lemma}\label{lemma:runaway}
When $n > \sum_{i=1}^k x_{i,1}$, there is exactly one handler $i$ with
more than $x_{i,1}$ objects allocated to it.
\end{lemma}
\begin{proof}

The following proof is via an exchange argument that exploits the structures of $F(x)$ and its gradient, as illustrated in Figure \ref{fig:identical_objects}. 

Since all objects are identical, we represent the allocation $S = \langle S_1, \dots, S_k \rangle$ of objects to $k$ handlers by $O = \langle o_1, \dots, o_k \rangle$, where $o_i = |S_i|$. Now, consider the allocation $O = \langle o_1, \dots, o_k\rangle$. Let
$o_i$ and $o_j$ be any two handlers with greater than $x_{i,1}$ and
$x_{j,1}$ allocated objects, respectively. We have $ES_i + ES_j =
F_i(x_{i,1} + \alpha) + F_j(x_{j,1} + \beta)$. Without loss of
generality, we assume that $p_i \geq p_j$.  Let $O^*$ be the
allocation resulting from a move of $1$ object from handler $i$ to
handler $j$. We have $ES^* _i + ES^* _j = F_i(x_{i,1} + \alpha -1) +
F_j(x_{j,1} + \beta + 1)$. The change in expected survival value is
given by:
\begin{align}
\left[ F_i(x_{i,1} + \alpha - 1) - F_i(x_{i,1} + \alpha) \right] + \left[ F_j(x_{j,1} + \beta + 1) - F_j(x_{j,1} + \beta)\right]\nonumber\\
= F'_j(x_{j,1} + \beta) - F'_i(x_{i,1} + \alpha - 1) \label{eq:diffp_runaway1}
\end{align}

Let $\widehat{O}$ be the allocation resulting from a move of $1$
object from handler $j$ to handler $i$. We have $\widehat{ES} _i +
\widehat{ES} _j = F_i(x_{i,1}+ \alpha + 1) + F_j(x_{j,1} + \beta
-1)$. The change in expected survival value is given by:
\begin{align}
\left[ F_i(x_{i,1} + \alpha + 1) - F_i(x_{i,1} + \alpha) \right] + \left[ F_j(x_{j,1} + \beta -1) - F_j(x_{j,1} + \beta)\right]\nonumber \\ 
= F'_i(x_{i,1} + \alpha) - F'_j(x_{j,1} + \beta - 1)\label{eq:diffp_runaway2}
\end{align}
Observe (in Figure \ref{fig:identical_objects}) that $F_i'(x)$ and $F'_j(x)$ are strictly increasing
functions for any input $x$ greater than $x_{i,2}$ and $x_{j,2}$,
respectively, and there is some point $x\in \mathbf{R}^+$ where $F'_i$
and $F'_j$ intersect -- i.e., $F'_j(x - \epsilon) < F'_i(x -
\epsilon)$ and $F'_j(x + \epsilon) > F'_i(x + \epsilon)$. Therefore,
for any $\beta$, there exists some $\gamma$ such that: $F'_j(x_{j,1} +
\beta) > F'_i(x_{i,1} + \gamma)$ and $F'_j(x_{j,1} + \beta) <
F'_i(x_{i,1} + \gamma - 1)$.
\begin{itemize}
\item If $\alpha - 1 > \gamma$: (\ref{eq:diffp_runaway1}) $> 0$
  -- therefore, we can simply move $1$ object from handler $i$ to
  handler $j$ and obtain a better solution than $O$.
\item If $\alpha - 1 \leq \gamma$: (\ref{eq:diffp_runaway2}) $>
  0$ -- therefore, we can simply move $1$ object from handler $j$ to
  $i$ and obtain a better solution than $O$.
\end{itemize}
Since we are able to make an improvement on any allocation containing
two or more handlers with more than $x_{*,1}$ allocated objects, we
have shown that no such solution may be optimal. Therefore, the
optimal solution has exactly one handler $i$ with more than $x_{i,1}$
objects.
\end{proof}

\begin{corollary}\label{cor:runaway}
When $n > \sum_{j=1}^k x_{j,1}$, if $|S_i| > x_{i,1}$ then $|S_i| = n -
(\sum_{j \neq i} x_{j,1})$ and $|S_j| = x_{j,1}$, $\forall j \in \{1,
\dots k\} \setminus \{i\}$.
\end{corollary}

\begin{algorithm}[h]
\caption{The Identical Objects DOH Algorithm when $n > \sum_{i=1}^k x_{i,1}$}
\label{alg:hdp_DOH}
\begin{algorithmic}
\Function{IO-DOH}{$n, k, p_1, \dots, p_k, x_{1,1}, \dots, x_{k,1}$}
	\For{$i = 1 \to k$}
		\State $S_i \gets \{\sum_{j=1}^{i-1} x_{j,1}, \dots, \sum_{j=1}^{i} x_{j,1}\}$
	\EndFor
	\For{$i = 1 \to k$}
		\State $S_i \gets S_i \bigcup \{\sum_{j=1}^{k} x_{j,1}, \dots, n\}$
		\If{$\max < ES\langle S_1, \dots, S_k \rangle$}
			\State $\max \gets ES\langle S_1, \dots, S_k \rangle$ and $OPT \gets \langle S_1, \dots, S_k \rangle$
		\EndIf
	\EndFor
\State \Return $OPT$ 
\EndFunction
\end{algorithmic}
\end{algorithm}


In the case where $n \leq \sum_{i=1}^k x_{i,1}$, algorithm
\ref{alg:hdp_DOH} runs the MMR algorithm to find an allocation. By
lemma \ref{lemma:MMR}, this is optimal.

In the case where $n > \sum_{i=1}^k x_{i,1}$, algorithm
\ref{alg:hdp_DOH} goes through $k$ iterations. In the $i^{th}$
iteration, it allocates $x_{j,1}$ objects to the $j^{th}$ handler
($\forall j \in \{1, \dots, k\}\setminus \{i\})$ and the remaining
objects to the $i^{th}$ handler. Following this, the $ES$ of the
allocation is computed and stored in the $i^{th}$ index of an
array. The algorithm returns the allocation corresponding to the
maximum value in this array as the optimal allocation. By lemma
\ref{lemma:runaway} and corollary \ref{cor:runaway}, this is the
optimal allocation.

\subsection{Identical Handlers and Values} \label{subsec:special_cases:identical_handlers_values}
In this case, we require that each object has the same value and
possesses failure probabilities that are object dependent (but,
handler independent) -- i.e., \emph{w.l.o.g.,} we have $v_j = 1$ and
$p_{ij} = p_j$ ($\forall i \in \{1, \dots, k\}$ and $\forall j \in
\{1, \dots, n\}$). Now the problem can be stated as:
\begin{problem}\label{problem:IVOIH_DOH}
The Identical Handlers and Values DOH (IHV-DOH) problem: 
Given $p_{ij} = p_j$ and $v_j = 1$ for $i=1,\dots k$ and
  $j=1\dots,n$, find partition $S_1\cup\dots\cup S_k$ of
  $\{1,\dots,n\}$ that maximizes $ES=\sum_{i=1}^kES_i$,
  where $ES_i=|S_i|\prod_{j\in S_i}p_j$.
\end{problem}
Observe that if instead, not all objects had identical values, the
problem would be $\mathsf{NP}$-$\mathsf{complete}$ as shown in theorem
\ref{th:np_completeness}.

\begin{theorem}\label{th:ivoih_doh}
Assuming that objects are ordered by decreasing survival probabilities
($p_j$), each handler receives a contiguous subsequence of objects.
\end{theorem}
\begin{proof}
Without loss of generality, let $p_1 \geq p_2 \geq \dots \geq
p_n$. Consider the allocation given by $S = \langle S_1, \dots, S_k \rangle$
where $\forall i \in \{1, \dots, k-1\}$: $ES_{i} \geq ES_{{i+1}}$. Let
$S_i$ be the first handler with a non-contiguous allocation -- i.e.,
object $l \in S_i$ but object $m \in S_{i+1}$ (where $p_l \leq
p_m$). Now, consider the allocation given by $S^* = \langle S_1, \dots,
S^*_i, S^*_{i+1}, \dots, S_k \rangle$ where object $l \in S^*_{i+1}$ and
object $m \in S^*_i$ -- i.e., the allocation where the handlers of
object $l$ and $m$ are swapped. Now, observe that: $\left[ES_{i} +
  ES_{{i+1}}\right] - \left[ES^*_{i} +
  ES^*_{{i+1}}\right]$=$\left[ES_{i} + ES_{{i+1}}\right] -
\left[\frac{p_m}{p_l}ES_{i} +
  \frac{p_l}{p_m}ES_{{i+1}}\right]$\\ $\leq \left[ES_{i} +
  ES_{{i+1}}\right] - \left[ES_{i} + ES_{{i+1}}\left(\frac{p_l}{p_m}
  +\frac{p_m}{p_l} - 1 \right)\right] \leq 0$. Therefore, the expected
survival value of a contiguous allocation is always at-least as good
as any non-contiguous allocation.
\end{proof}

Since the optimal allocation of objects to handlers is contiguous
(given objects sorted by their survival probabilities), the following
recursive relation may be used to find the optimal allocation.
\begin{equation*}
OPT(i,j) = \max_{1 \leq l \leq i} \left[ \left( i - l\right)\prod_{m=l}^i p_m  + OPT(l,j-1)\right]
\end{equation*}
where $OPT(i,0) = \infty$ and $OPT(n,k)$ returns the maximum expected
survival value for $n$ objects allocated to $k$ handlers in $O(n^2k)$
time.

\subsection{Identical Risks} \label{subsec:special_cases:identical_risks}
In this case, we require that each object possesses a distinct value
and failure probabilities that are handler dependent (but, object
independent) -- i.e., \emph{w.l.o.g.,} we have $p_{ij} = p_i$
($\forall i \in \{1, \dots, k\}$ and $\forall j \in \{1, \dots,
n\}$). Now the problem can be stated as the following:
\begin{problem}\label{problem:ERDH_DOH}
The Identical Risk DOH (IR-DOH) problem:
Given $p_{ij} = p_i$ and for $i=1,\dots k$ and
  $j=1\dots,n$, find partition $S_1\cup\dots\cup S_k$ of
  $\{1,\dots,n\}$ that maximizes $ES=\sum_{i=1}^kES_i$,
  where $ES_i=\sum_{j\in S_i}v_j p_i^{|S_i|}$.
\end{problem}

\begin{theorem}\label{th:erdh_doh}
Assuming that objects are ordered by decreasing values ($v_j$) and
handlers are ordered by decreasing survival probabilities($p_i$), each
handler receives a contiguous subsequence of objects.
\end{theorem}
\begin{proof}
Let $v_1 \geq v_2 \geq \dots \geq v_n$. Consider the allocation given
by $S = \langle S_1, \dots, S_k \rangle$ where without loss of generality,
$p_1^{|S_1|} \geq p_2^{|S_2|} \geq \dots \geq p_k^{|S_k|}$. Let $S_i$
be the first handler with a non-contiguous allocation -- i.e., object
$l \in S_i$ but object $m \in S_{i+1}$ (where $v_l \leq v_m$). Now,
consider the allocation given by $S^* = \langle S_1, \dots, S^*_i,
S^*_{i+1}, \dots, S_k \rangle$ where object $l \in S^*_{i+1}$ and object $m
\in S^*_i$ -- i.e., the allocation where the handlers of object $l$
and $m$ are swapped. Now, observe that:
\begin{align*}
\left[ES_{i} + ES_{{i+1}}\right]- \left[ES^*_{i} + ES^*_{{i+1}}\right] = \\
\left[p_i^{|S_i|}V_{i} + p_{i+1}^{|S_{i+1}|}V_{{i+1}} \right]-\left[p_i^{|S_i|}(V_{i} + v_m - v_l) + p_{i+1}^{|S_{i+1}|}(V_{{i+1}} + v_l - v_m) \right] = \\  
\left[p_{i}^{|S_{i}|}(v_m - v_l) - p_{i+1}^{|S_{i+1}|}(v_l - v_m) \right] \leq 0. 
\end{align*}

Therefore, the expected survival value of a contiguous allocation is
always at-least as good as any non-contiguous allocation.
\end{proof}

Since the optimal allocation of objects to handlers is contiguous
(given objects sorted by their values), the following recursive
relation may be used to find the optimal allocation.
\begin{equation}\label{eq:rec_sp_case1}
OPT(i,j) = \max_{1 \leq l \leq i} \left[  p_j^{i - l} \sum_{m=l}^i v_m + OPT(l,j-1)\right]
\end{equation}
where $OPT(i,0) = \infty$ and $OPT(n,k)$ returns the maximum expected
survival value for $n$ objects allocated to $k$ handlers in $O(n^2k)$
time.

\subsection{Identical Risks, Values and Handlers}\label{subsec:special_cases:all_identical}
Under the assumption that all objects and handlers are identical, \emph{w.l.o.g.,} we have $v_j = 1$ and $p_{ij} = p$ ($\forall i \in \{1, \dots, k\}$ and $\forall j \in \{1, \dots, n\}$). In this case, the generic \emph{DOH} problem may be re-stated as:

\begin{problem}\label{problem:ae_DOH}
The Identical Objects and Identical Handlers DOH (IOIH-DOH) problem: 
Given $p_{ij} = p$ and $v_j = 1$ for $i=1,\dots k$ and
  $j=1\dots,n$, find partition $S_1\cup\dots\cup S_k$ of
  $\{1,\dots,n\}$ that maximizes $ES=\sum_{i=1}^kES_i$,
  where $ES_i=|S_i|p^{|S_i|}$.
\end{problem}

We make the following observations about $F(x) = xp^x$, where $x \in \mathbf{Z} ^+$:
\begin{itemize}
\item[O1:] For any $p$ such that $0<p<1$, there exists an $x_1$ such that $F(x_1) \geq F(x_1 + \alpha)$ and $F(x_1) \geq F(x_1 - \alpha)$, where $\alpha \in \mathbf{Z} ^+$ and there exists an $x_2$ such that $F'(x_2) \leq F'(x_2 + \beta)$ and $F'(x_2) \leq F'(x_2 - \beta)$, where $\beta \in \mathbf{Z} ^+$.
\item[O2:] $\forall x \in \{0, \dots, x_1-1\}$: $F(x) > 0$, $F'(x) > 0$, and $F''(x) < 0$.
\item[O3:] $\forall x \in \{x_1+1, \dots, \infty\}$: $F(x) > 0$ and $F'(x) \leq 0$.
\item[O4:]$\forall x \in \{x_1+1, \dots, x_2\}$: $F''(x) \leq 0$ and $\forall x \in \{x_2+1, \dots, \infty\}$: $F''(x) \geq 0$.
\end{itemize}

\begin{algorithm}
\caption{IOIH-DOH Solver}
\label{alg:ae_DOH}
\begin{algorithmic}
\If{$x_1 > n$}
	\State $x_1 \gets n$.
\EndIf
\If{$n \leq kx_1$}
	\State $OPT \gets \langle \lfloor \frac{n}{k} \rfloor, \dots, \lfloor \frac{n}{k} \rfloor, \lceil \frac{n}{k} \rceil, \dots, \lceil \frac{n}{k} \rceil \rangle$ such that $\sum_{i=1}^k OPT_i = n$.
\EndIf

\If{$n > kx_1$}
	\State $O \gets \langle x_1, \dots, x_1, n-(k-1)x_1\rangle$
	\For{$i=1 \to k-1$}
		\If{$\left[F(O_k-1) - F(O_k)\right] > \left[F(O_i) - F(O_i + 1)\right]$}
			\State $O_k \gets O_k-1$
			\State $O_i \gets O_i + 1$
		\EndIf
	\EndFor
	\State $OPT \gets O$
\EndIf
\State \Return $OPT$
\end{algorithmic}
\end{algorithm}

\begin{lemma}\label{lemma:ae_DOH_1}
In any optimal allocation for problem \ref{problem:ae_DOH} with $n\leq
kx_1$, the objects are split amongst handlers as evenly as possible.
\end{lemma}
\begin{proof}
Consider the allocation\footnote{The vector $O = \langle o_1, \dots,
  o_k\rangle$ denotes an allocation where $o_i$ is the size of the set
  containing all the objects allocated to the $i^{th}$ handler --
  i.e., $o_i = |S_i|$.} $O = \langle o_1, \dots, o_k \rangle$ with
$|o_i| > x_1$ and $|o_j| < x_1$. We have $ES_i + ES_j = F(x_1 +
\alpha) + F(x_1 - \beta)$. By moving a single object from $o_i$ to
$o_j$, we get the allocation $O^*$ where $ES^*_i + ES^*_j = F(x_1 +
\alpha - 1) + F(x_1 + \beta + 1)$. By O2 and O3, $[ES^*_i - ES_i] +
       [ES^*_j - ES_j] > 0$. Therefore, $O^*$ is strictly better than
       $O$.
\end{proof}

\begin{lemma}\label{lemma:ae_DOH_2}
In any optimal allocation for problem \ref{problem:ae_DOH} with $n >
kx_1$, exactly $1$ handler is allocated more than $x_1 + 1$ objects.
\end{lemma}
\begin{proof}
Consider the allocation $O = \langle o_1, \dots, o_k \rangle$ with
$|o_i |\geq |o_j| > x_1 + 1$. We have $ES_i + ES_j = F(x_1 + \alpha) +
F(x_1 + 1 + \beta)$. If $F(x_1) - F(x_1+1) < F(x_1 + \alpha + \beta) -
F(x_1 + \alpha + \beta + 1)$, then we move $\beta$ objects from $o_j$
to $o_i$ to get the allocation $O^*$ where $ES^*_i + ES^*_j = F(x_1 +
\alpha + \beta) + F(x_1 + 1)$. Otherwise, we move $\beta + 1$ objects
from $o_j$ to $o_i$ to get the allocation $\widehat{O}$ where
$\widehat{ES}_i + \widehat{ES}_j = F(x_1 + \alpha + \beta + 1) +
F(x_1)$.  By O3 and O4, we have either $[ES^*_i - ES_i] + [ES^*_j -
  ES_j] > 0$ or $[\widehat{ES}_i - ES_i] + [\widehat{ES}_j - ES_j] >
0$. Therefore, $O$ can never be optimal.

Note: If $x_1 \in \mathbf{R}^+$ (rather than $x_1 \in \mathbf{Z}^+$),
it is always true that $F(x_1) - F(x_1+1) > F(x_1 + \alpha + \beta) -
F(x_1 + \alpha + \beta + 1)$. Therefore, the optimal allocation would
always have exactly $k-1$ handlers with exactly $x_1$ objects
allocated to each of them.
\end{proof}

\begin{theorem}\label{th:ae_DOH_optimal}
Algorithm \ref{alg:ae_DOH} always finds an optimal allocation for problem \ref{problem:ae_DOH}.
\end{theorem}
\begin{proof}
In the case where $n \leq kx_1$, algorithm \ref{alg:ae_DOH} returns
the allocation where objects are split as evenly as possible amongst
the handlers. By lemma \ref{lemma:ae_DOH_1}, this is optimal.

In the case where $n > kx_1$, algorithm \ref{alg:ae_DOH} first
allocates exactly $x_1$ objects to each of $k-1$ handlers and the
remaining objects to the $k^{th}$ handler. Then, if we have $F(x_1) -
F(x_1 + 1) < F(O_k -1) - F(O_k)$, we move one object from the $k^{th}$
handler to one of the other $k-1$. This shifting process is repeated
until the condition is no longer true (at-most $k-1$ iterations). By
lemma \ref{lemma:ae_DOH_2}, this is optimal.
\end{proof}

\subsection{Exactly One Object per Handler}\label{subsec:special_cases:EOVA}
In this section we consider the original formulation defined in
problem \ref{problem:DOH}, this time with the restriction that
handlers may be allocated only exactly one object. We will show that
under this assumption, the problem can be stated as an instance of the
transportation problem (a special case of the min cost - max flow
problem with a linear objective function) \cite{tp_3}. This special
case can be stated as the following:
\begin{problem}\label{problem:exactone}
The Exactly One DOH (EO-DOH) problem:
 Given $v_1,\dots,v_n$ and $p_{ij}$ for $i=1,\dots k$ and
  $j=1\dots,n$, find partition $S_1\cup\dots\cup S_k$ of
  $\{1,\dots,n\}$ that maximizes $ES=\sum_{i=1}^kES_i$,
  where $ES_i=\sum_{j\in S_i}v_j\prod_{j\in S_i} p_{ij}$
  and $|S_i| = 1$ for $i = 1, \dots, k$.
\end{problem}

Since $|S_i| = 1$, the following identities are observed: (1)
$\prod_{j \in S_i}p_{ij} = \sum_{j \in S_i} p_{ij}$ and (2) $(\sum_{j
  \in S_i} p_{ij}) (\sum_{j \in S_i}v_j) = \sum_{j \in
  S_i}p_{ij}v_j$. These identities allow us to convert problem
\ref{problem:exactone} into the instance of an (unbalanced)
transportation problem \cite{tp_1}, \cite{tp_2} defined in problem
\ref{problem:Transportation}.
\begin{problem}\label{problem:Transportation}
The Unbalanced Transportation Problem:
 Given $v_1,\dots,v_n$ and $p_{ij}$ for $i=1,\dots k$ and
  $j=1\dots,n$, find partition $S_1\cup\dots\cup S_k$ of
  $\{1,\dots,n\}$ that maximizes $ES=\sum_{i=1}^kES_i$,
  where $ES_i=\sum_{j\in S_i}v_j p_{ij}$.
\end{problem}

Problem \ref{problem:Transportation} is a well studied problem with
efficient methods for obtaining solutions based on successive
shortest-path generalizations \cite{nf_1_alg}, minimum mean cycle
cancelling \cite{nf_3_alg}, and poly-time network simplex
methods\cite{nf_2_alg}.

\section{Heuristics}\label{sec:heuristics}
In this section we propose and evaluate five heuristics for finding solutions to DOH problems. The heuristics are based on the MMR algorithm (Algorithm \ref{alg:mmr}), dynamic programming approach (Section \ref{sec:approximations}), and genetic algoritms. A performance comparison of the heuristics is provided in Table \ref{tab:heuristics}. 

\paragraph{Ordered MMR for DOH Problems (O-MMR): }In this heuristic, objects are ordered by decreasing values before the standard MMR greedy approach (Section \ref{sec:special_cases}) is used as usual. Experimental analysis reveals that in most cases, such an ordering (usually) performs better than MMR where objects are either unordered, or ordered by increasing values.

\paragraph{Clairvoyant MMR for DOH Problems (C-MMR): }The following variation is made to the O-MMR algorithm: If there is no handler to which the current object can be allocated without causing a drop in the cumulative expected survival value, then that object is placed in a {dumpster} handler. After initial allocation of all objects is complete, all objects in the dumpster are allocated together (as a single object) to the one handler that experiences the smallest loss in ES value.

\begin{algorithm}[h]
\caption{The MMR, OMMR, and CMMR Algorithms} \label{alg:mmr}
\small
\begin{algorithmic}
\Function{MMR}{$n,k,v_1, \dots, v_n, p_{11}, \dots, p_{kn},type$}\Comment{$type \gets 1$ for OMMR, $2$ for CMMR, $0$ otherwise}
\For{$i = 1 \to k$}
	\State $S_i \gets \emptyset$, $ES_i \gets 0$, $S_{dumpster} \gets \emptyset$
\EndFor
\If{$type \geq 1$}
	\State $\Call{sort}{v_1, \dots, v_n, p_{11}, \dots, p{kn}}$\Comment{sort values and re-arrange probabilities accordingly}
\EndIf
\For{$i = 1 \to n$}
	\For{$j = 1 \to k$}
		\State $T_j \gets S_j \bigcup i$, $\delta_j \gets \Call{Compute-ES}{T_j} - ES_j$
	\EndFor
	\State $m \gets \arg\max \{\delta_1, \dots, \delta_k\}$
	\If {$cmmr = 1 \land \delta_m < 0$}
		\State $S_{dumpster} \gets S_{dumpster} \bigcup i$
	\Else
		\State $S_m \gets S_m \bigcup i$, $ES_m \gets \Call{Compute-ES}{S_m}$
	\EndIf
\EndFor
\If{$type = 2$}
\For{$j = 1 \to k$}
	\State $T_j \gets S_j \bigcup S_{dumpster}$, $\delta_j \gets \Call{Compute-ES}{T_j} - ES_j$
\EndFor
\State $m \gets \arg\max \{\delta_1, \dots, \delta_k\}$
\EndIf
\State \Return $\langle S_1, \dots, S_k \rangle$
\EndFunction
\end{algorithmic}
\end{algorithm}

\paragraph{Dynamic Programming Based Heuristic (DP-H): }The DP-H algorithm is a variation of the FPTAS presented in Section \ref{sec:approximations}. The major difference is that the DP-H algorithm maintains exactly $n$ states for each iteration of the dynamic program. In each iteration, the $n$ {most promising} states are maintained while the remaining are culled. To do this, we divide the state space into $n$ uniformly sized blocks and maintain at-most $n$ states for each iteration of the dynamic program -- i.e., the state with the highest ES for each block. Therefore, in each iteration, no more than $n$ \emph{promising} states are maintained while the remaining are culled. Finally, after all $n$ accounts are allocated, the state with the maximum ES is returned. The algorithm is illustrated in algorithm \ref{alg:dp_h}.

\begin{algorithm}
  \caption{DP-H trimming procedure}
  \label{alg:dp_h}
  \begin{algorithmic}
    \Function{Trim}{$\Phi$} \Comment{\textsc{Trim} function for computing DP-H solutions}
	    \ForAll{$(V_1,\dots,V_k,P_1,\dots,P_k) \in \Phi$}
	      \State Find $\langle {r_1}, \dots, {r_{k}} \rangle$ such that $\frac{(r_i+1)V_{max}}{n^{1/k}}\leq V_i \leq \frac{r_iV_{max}}{n^{1/k}}$, $\forall i \in \{1, \dots, k\}$
	      \State Find $\langle s_{1}, \dots, {s_{k}} \rangle$ such that $\frac{(s_i+1)P_{max}}{n^{1/k}}\leq P_i \leq \frac{s_iP_{max}}{n^{1/k}}$, $\forall i \in \{1, \dots, k\}$
	      \State $CurrES \gets \Call{Compute-ES}{V_1,\dots,V_k,P_1,\dots,P_k}$
	      \If{$Orthotope[\langle r_1, \dots, r_{k}, s_1,\dots,s_k \rangle] = \perp$}
		      \State $Orthotope[\langle r_1, \dots, r_{k}, s_1,\dots,s_k \rangle] \gets (V_1,\dots,V_k,P_1,\dots,P_k)$
		      \State $ESofOrthotope[\langle r_1, \dots, r_{k}, s_1,\dots,s_k \rangle] \gets CurrES$
		  \EndIf
		  
		  \If{$ESofOrthotope[\langle r_1, \dots, r_{k}, s_1,\dots,s_k \rangle] < CurrES$}
		  		      \State $Orthotope[\langle r_1, \dots, r_{k}, s_1,\dots,s_k \rangle] \gets (V_1,\dots,V_k,P_1,\dots,P_k)$
		  		      \State $ESofOrthotope[\langle r_1, \dots, r_{k}, s_1,\dots,s_k \rangle] \gets CurrES$
	      \EndIf
      \EndFor
      \State \Return $Orthotope$
    \EndFunction
  \end{algorithmic}
\end{algorithm}

\paragraph{Random Initialization with Complete Genetic Inheritance (RI-GA): }We refer to each DOH candidate solution as an individual and each single object-to-handler assignment as a chromosome of that individual. The fitness function used to evaluate an individual is simply function that computes its corresponding ES value. In each iteration, the fitness of all individuals is evaluated and a fixed number of {couples} are selected based on the standard roulette wheel selection \cite{RWS}. After the mutation procedure, each couple is replaced with their offspring.

In the RI-GA algorithm, the initial population is randomly generated. The child $(c)$ of each couple $(p_1, p_2)$ is created by inheriting every chromosome from one of $p_1$ or $p_2$. The most profitable non-overlapping handlers (set of chromosomes) of $p_1$ and $p_2$ are inherited in decreasing order. Finally, if no more chromosome sets may be completely inherited and the child is still incomplete, the child inherits each of its remaining chromosomes from parent $p_i$ with probability $\frac{fitness(p_i)}{(fitness(p_1) + fitness(p_2))}$. 

\paragraph{Heuristic Based Initialization with Partial Genetic Inheritance (HI-GA): }In the HI-GA algorithm, the initial population is seeded as follows: $60\%$ are randomly generated, $15\%$ are populated by the allocation generated by O-MMR, $15\%$ are populated by the allocation generated by C-MMR, and the final $10\%$ are populated by the allocation generated by DP-H. The mutation procedure is similar to RI-GA, except if no more chromosome sets may be completely inherited and the child is still incomplete, the child spawns randomly generated chromosomes (which may completely differ from either parent).

In the evaluation of the RI-GA and HI-GA algorithms, the initial population was set to be 1000. Further, in each generation, 1000 pairs of individuals were chosen to produce 1000 children (the population was always exactly 1000). The fittest child after 1000 generations was returned as the solution.
\begin{table}
\centering
\begin{tabular}{|c|c|c|c|c|c|c|c|c|c|c|c|}
\hline
\multirow{2}{*}{\emph{n}} & \multirow{2}{*}{{\emph{k}}} & \multicolumn{2}{c|}{{O-MMR}} & \multicolumn{2}{c|}{{C-MMR}} & \multicolumn{2}{c|}{{DP-H}}& \multicolumn{2}{c|}{{RI-GA}}& \multicolumn{2}{c|}{{HI-GA}}\\ 
\cline{3-12}
 & & $\mu$ & $\sigma$ & $\mu$ & $\sigma$ &$\mu$ & $\sigma$ &$\mu$ & $\sigma$ &$\mu$ & $\sigma$ \\
 \hline
 \multirow{2}{*}{25} & 5 & .398 & .070 & .380 & .066 & .398 & .069 & .448 & .061 & .464 & .056 \\
 \cline{2-12}
 & 10 & .759 & .044 & .767 & .050 & .759 & .045 & .705 & .048 & .767 & .044 \\
 \hline
 \multirow{2}{*}{50} & 10 & .585 & .047 & .561 & .046 & .585 & .046 & .459 & .047 & .588 & .044 \\
 \cline{2-12}
 & 25 & .908 & .013 & .910 & .013 & .908 & .015 & .685 & .038 & .910 & .036 \\
 \hline
 \multirow{2}{*}{100} & 25 & .832 & .019 & .829 & .024 & .831 & .020 & .445 & .043 & .831 & .019 \\
 \cline{2-12}
 & 50 & .950 & .006 & .950 & .006 & .941 & .006 & .559 & .030 & .950 & .006 \\
 \hline
 \multirow{3}{*}{250} & 25 & .654 & .021 & .639 & .022 & .654 & .021 & .121 & .029 & .681 & .028 \\
 \cline{2-12}
 & 50 & .890 & .008 & .886 & .009 & .889 & .008 & .221 & .027 & .891 & .008 \\
 \cline{2-12}
 & 100 & .969 & .002 & .970 & .002 & .969 & .002 & .348 & .022 & .970 & .022 \\
\hline
\end{tabular}
\caption{Mean ($\mu$) and Standard Deviation ($\sigma$) of the ratio between the ES obtained by the heuristic and the upper-bound ES value for varying $n$ and $k$ (50 trials per experiment).}
\vspace{-.45in}
\label{tab:heuristics}
\end{table}

\section{Conclusions and Future Work}\label{sec:conclusions}

The Destructive Object Handler problem describes many real-world
situations in which dangerous objects must be shipped, processed,
quarantined, or otherwise handled, and the destruction of one object
assigned to a handler destroys all the other objects assigned to the
same handler.  We have shown that this problem is NP-complete in
general, but have provided an FPTAS when the number of handlers is
constant and polynomial time algorithms for numerous special cases.  We
have evaluated several heuristics based on simple greedy strategies, dynamic programming, and genetic algorithms. It appears that the heuristics perform poorly when the ratio of objects to handlers is high, but improve as this ratio reduces (even for very large problem sizes).

There remain many open avenues for research on the DOH problem, particularly in understanding its complexity. The hardness of the generic {DOH} problem is still not fully
known -- i.e., it remains unknown if the {DOH} problem is
solvable by a pseudo-polynomial time algorithm or if there exists a
reduction that proves its strong
$\mathsf{NP}$-$\mathsf{Completeness}$.  We also have not studied the
special case when all objects have identical values.  Further, it is
not known if efficient constant factor or $(1- \epsilon)$
approximations exist when $k$ is not a constant.

\bibliographystyle{splncs} 
\bibliography{optimalbib}

\appendix

\end{document}